%% file: main_arxiv.tex
\documentclass[sigconf,natbib=false]{acmart}
\AtBeginDocument{%
  }

\setcopyright{acmlicensed}
\copyrightyear{2026}
\acmYear{2026}
\acmDOI{XXXXXXX.XXXXXXX}
\acmConference[CCS '26]{Conference on Computer and Communications Security}{November 15--19,
  2026}{The Hague, The Netherlands}
\acmISBN{978-1-4503-XXXX-X/18/06}

\RequirePackage[
  datamodel=acmdatamodel,
  style=acmnumeric,
  ]{biblatex}

\DeclareSourcemap{
    \maps[datatype=bibtex, overwrite]{
        \map{
            \step[fieldsource=url, match=\regexp{\x{5F}}, replace=\regexp{\\\x{5F}}]
        }
    }
}

\input{macros}

\usepackage{tabularx}

\usepackage{pgfplots}
\usepackage{framed}
\usepackage{tikz}
\usepackage{svg}
\usepackage{multirow}
\usepgfplotslibrary{fillbetween,groupplots}
\usepackage{mathtools}
\usepackage{subcaption}
\usetikzlibrary{calc}
\usetikzlibrary{patterns.meta}
\usetikzlibrary{decorations.text}
\usepackage[underline=true]{pgf-umlsd}

\begin{document}

\title{Time-Delayed Publicly Verifiable Quantum Computation for Classical Verifiers}

\author{Ameer Mohammed}
\authornotemark[1]
\email{ameer.mohammed@ku.edu.kw}
\affiliation{%
  \institution{Kuwait University}
  \city{Kuwait City}
  \state{}
  \country{Kuwait}
}
\author{Aydin Abadi}
\email{aydin.abadi@ncl.ac.uk}
\authornotemark[1]
\affiliation{%
  \institution{Newcastle University}
  \city{Newcastle upon Tyne}
  \state{}
  \country{United Kingdom}
}
\author{Jaffer Mahdi}
\email{jaffer.mahdi@ku.edu.kw}
\authornotemark[1]
\affiliation{%
  \institution{Kuwait Oil Company}
  \city{Kuwait City}
  \state{}
  \country{Kuwait}
}

\begin{abstract}

\input{sections/00_abstract}
\end{abstract}



\maketitle

\input{sections/01_intro}
\input{sections/02_related}
\input{sections/03_prelims}
\input{sections/04_protocol}

\input{sections/05_experimental}
\input{sections/06_conclusion}
\input{sections/08_acknowledgments}

\printbibliography

\appendix
\section{Open Science Appendix}
The code used containing the implementation of the protocol used to obtain our experimental results and instructions to run it can be found at~\cite{QuantumVerificationTLP_Code}.



\input{sections/Appndx_commitment_def}

\end{document}

%% file: macros.tex
\usepackage{amsfonts}

\usepackage{amsmath}

\usepackage{enumitem}

\usepackage{xspace}

\newcommand{\secp}{\mathrm{\lambda}}

\usepackage{graphicx}

\newcommand{\remove}[1]{}

\newcommand{\etal}{et~al.\ }

\newcommand{\comcom}{\ensuremath{\mathtt{Com}}\xspace}
\newcommand{\comver}{\ensuremath{\mathtt{Ver}}\xspace}

\newcommand{\angles}[1]{\langle #1 \rangle}

\newcommand{\ket}[1]{\lvert#1\rangle}

\renewcommand{\L}{\ensuremath L\:}

\newcommand{\N}{{\mathbb N}}

\newcommand{\cB}{{\mathcal B}}
\newcommand{\cC}{{\mathcal C}}

\newcommand{\cF}{{\mathcal F}}

\newcommand{\cH}{{\mathcal H}}

\usepackage{bm}

\newcommand{\poly}{\operatorname{poly}}

\newcommand{\negl}{\operatorname{negl}}

\newcommand{\class}[1]{\ensuremath{\mathbf{#1}}}

\newcommand{\BPP}{\class{BPP}}

\newcommand{\NP}{\ensuremath{\class{NP}}}

\renewcommand{\P}{\class{P}}

\newcommand{\BQP}{\class{BQP}}
\newcommand{\QMA}{\class{QMA}}

\newcommand{\QPIP}{\class{QPIP}}

\newtheorem{theorem}{Theorem}[section]

\newtheorem{claim}[theorem]{Claim}
\newtheorem{remark}[theorem]{Remark}

\newtheorem{proto}[theorem]{Protocol}
\newtheorem{protoc}[theorem]{Protocol}

\newcommand{\namedref}[2]{#1~\ref{#2}}

\newcommand{\torestate}[3]{%
\expandafter \def \csname BBRESTATE #2 \endcsname{#3}
\theoremstyle{plain}
\newtheorem{BBRESTATETHMNUM#2}[theorem]{#1}
\begin{BBRESTATETHMNUM#2}\label{#2}\csname BBRESTATE #2 \endcsname   \end{BBRESTATETHMNUM#2}
\newtheorem*{BBRESTATETHMNONNUM#2}{\namedref{#1}{#2}}
}

\newcommand{\restate}[1]{\begin{BBRESTATETHMNONNUM#1}[Restated] \csname BBRESTATE #1 \endcsname
\end{BBRESTATETHMNONNUM#1}}

\newcommand{\TLP}{\mathsf{TLP}}
\newcommand{\TLPSetup}{\mathsf{TLP.Setup}}
\newcommand{\TLPGen}{\mathsf{TLP.GenPuzzle}}
\newcommand{\TLPSolve}{\mathsf{TLP.Solve}}

\newcommand{\VC}{\mathsf{VC}}
\newcommand{\VCSetup}{\mathsf{VC.Setup}}
\newcommand{\VCProve}{\mathsf{VC.Prove}}
\newcommand{\VCReveal}{\mathsf{VC.Reveal}}
\newcommand{\VCVerify}{\mathsf{VC.Verify}}

\newcommand{\Commit}{\mathsf{Commit}}

\newcommand\compind{\stackrel{\mathclap{\normalfont\mbox{c}}}{\approx}}

\newcommand\statind{\stackrel{\mathclap{\normalfont\mbox{s}}}{\approx}}

%% file: sections/00_abstract.tex


%

Publicly verifiable delegation is a well-known problem involving a user who wishes to outsource a resource-intensive computational task to a more powerful but potentially untrusted server such that any other party is able to efficiently check the veracity of the computation's result. This problem has been extensively studied in the classical domain where the user and server are both non-quantum machines. However, the problem becomes more challenging when the classical user wants to delegate a quantum circuit to a \textit{single} prover with quantum-computing capabilities. Previous solutions have resorted to using impractical or non-standard cryptographic solutions (e.g. indistinguishability obfuscation) to achieve this requirement. In this work, we relax the requirement to have \textit{time-delayed} publicly verifiable proofs, where the verification key is made known to the public only when the computation (and its proof) are guaranteed to have been completed. We propose a practical non-interactive scheme leveraging commitment schemes and time-lock puzzles, which can be efficiently realized through well-established and standard post-quantum assumptions. The main idea of our technique lies in using time-lock puzzles to compile a 2-round privately verifiable scheme into a non-interactive publicly verifiable scheme with timestamped proofs, outsourcing not only the quantum computation but the puzzle solving as well. Security is proven in the quantum random oracle model with a common reference string (CRS).

%% file: sections/01_intro.tex
\section{Introduction}

An important problem when outsourcing a client's computation to a more capable server is ensuring that the result that is computed and returned by the server is correct and efficiently verifiable by the client. In particular, the verification complexity should be much lower than that required to run the actual computation. This problem has been well studied when the server and client are both non-quantum machines, and, in fact, can be realized through the use of succinct arguments for \NP~\cite{kilian1992, micali1994}. 
With quantum computing on the horizon and likely to be offered initially as a pay-per-use service by major companies, it is vital to prepare for a scenario where users with classical computers want to delegate complex computations to powerful quantum servers and efficiently verify the results produced by the servers-- a task now known in the literature as Classical Verification of Quantum Computation (CVQC).

In the breakthrough work of Mahadev \cite{mahadev2018}, it was demonstrated that such a scheme does indeed exist assuming the post-quantum hardness of Learning With Errors (LWE). This was further extended \cite{chia2020, bartusek2022} to add succinctness, in which the verifier's time complexity is $\poly(\log T)$ where $T$ is the time to execute the delegated quantum computation. In addition, while the original protocol consisted of 4 rounds of communication, this was reduced to 2 rounds or less (1 round with setup) while still preserving negligible soundness through parallel repetition \cite{chia2020,alagic2020,bartusek2022}. 

More relevant to this work, some applications demand \textit{public verifiability}, an additional property allowing anyone (not just the delegator) to verify the result of the computation. 

The usual desired setting of publicly verifiable non-interactive verifiable computation involves a prover who wishes to publish a single-message proof convincing any verifier that $C(x) = y$ for some program or circuit $C$ over input $x$. To avoid triviality in the fully classical setting, the proof needs to be relatively ``short'' such that verification is faster than re-executing $C(x)$ and checking the result, otherwise the verifier has no need to delegate and can instead perform the computation itself. The proof size requirement, while preferred, is not highly prioritized when the prover is quantum, as $C$ is not necessarily efficiently computable by the verifier. 
%
%
The standard go-to approach to solve this problem is by using succinct non-interactive arguments (SNARGs) for $\NP$ (or even $\P$), which can be realized from LWE \cite{choudhuri2022}. However, our delegated computation is quantum, and thus one needs a SNARG for $\BQP$ for this approach to work. The work of Bartusek and Malavolta \cite{bartusek2022a} provides one means to achieve this goal, and, in fact, they obtain SNARGs for $\QMA$, the quantum analogue of $\NP$, (and hence $\BQP$) which would imply publicly verifiable CVQC. Their results rely on quantum null-iO, which is derived from non-standard cryptographic primitives, e.g., VBB obfuscation. In contrast, our goal is to prioritize practicality over theoretical feasibility, and therefore, we wish to avoid such strong assumptions.

Ideally, we would like to base the scheme on \textit{falsifiable assumptions} ~\cite{naor2003,gentry2011b}, which are those that can be tested in a security game against an efficient challenger. Examples include one-way functions, trapdoor permutations, DDH, and LWE. Notably, indistinguishability obfuscation is \textit{not} falsifiable \cite{jain2022}. Roughly speaking, such assumptions are often preferred, as they can be efficiently proven/refuted and therefore instill more confidence when used as building blocks for more advanced schemes. The very recent related work of Bartusek et al. \cite{bartusek2026} comes close to achieving this by showing a publicly verifiable CVQC that is based on LWE but they do so under the classical oracle model \cite{bartusek2023a}, which idealizes the notion of obfuscation for classical circuits and can be heuristically instantiated with iO.








\subsection{Our Result}

We seek to make progress toward addressing the following question: 
\begin{center}
\textit{Does a publicly verifiable delegation protocol for quantum computation with a classical verifier exist under \textbf{falsifiable} assumptions in the random oracle model?}
\end{center}
We answer the question affirmatively with a minor adjustment to the standard publicly verifiable computation setting. 
Unlike standard publicly verifiable computation, where the verification key is available from the outset, our scheme delays the release of this capability. Specifically, public verifiability is enabled only after a time delay $\Delta$, where $\Delta > O(T)$, and only for proofs generated before time $\Delta$. Until that point in time, the scheme operates in a privately verifiable mode. 

\begin{theorem}[informal] Under the post-quantum LWE assumption, there exists a non-interactive time-delayed publicly verifiable CVQC protocol for problems in \textbf{BQP} with computational soundness in the common reference string (CRS) and quantum random oracle model (QROM).
\end{theorem}

A scheme in the common
reference string (CRS) model, as opposed to the plain model, would allow all parties to have access to a string that was honestly generated by some trusted setup algorithm. We note that the CRS is not uniformly random but is deliberately structured,
so we cannot simply replace the need for a CRS with an output of the random oracle. 


Our technique of employing time-lock puzzles to render a CVQC protocol publicly verifiable is also of independent interest. To the best of our knowledge, this is the first instance in which time-lock puzzles are used to transform a designated-verifier verifiable quantum computation scheme into a publicly verifiable one. 


We expect that this technique can be generalized to compile any privately verifiable computation scheme into a time-delayed, publicly verifiable one, even in the classical (non-quantum) setting. However, since standard publicly verifiable delegation schemes based on falsifiable assumptions already exist in the classical setting~\cite{kalai2019,ghosal2023}, this generalization is more impactful and meaningful in the classical-verifier quantum-prover setting. 
We expect that the time-delayed publicly verifiable scheme offers greater efficiency in scenarios where lightweight alternatives to SNARGs are preferred.

We implemented an instantiation of our scheme for two computational tasks on AWS Braket, which is a cloud-based quantum simulation platform. We tested our scheme on random quantum circuits and Harrow–Hassidim–Lloyd (HHL) circuits \cite{harrow2009,yalovetzky2024}, and estimated their concrete running time. The latter was selected as one potential candidate circuit that a client may delegate to take advantage of the quantum speedup and solve linear systems exponentially faster than classical algorithms. In assessing concrete running times for a $16 \times 16$ HHL instance, we found that circuit computation by the quantum prover took 90 seconds (which is also the time it should take at least to solve the corresponding puzzle) while puzzle generation on the classical client side required only~0.05\,ms.

\input{sections/figure-workflow}
\subsection{Motivation and Applications}

Quantum computers promise significant speedups over classical machines for a wide range of tasks. However, due to their high cost and limited availability, they are expected to be initially accessed via cloud-based, pay-per-use services. In such a model, clients (often classical and resource-limited) need to outsource their quantum computations to powerful remote quantum servers. This delegation model introduces the critical challenge of trust:

\begin{center}
\textit{How can a classical client be sure that the result provided by a remote quantum server is correct? }
\end{center}

Verifiable delegation schemes aim to solve this by enabling the client to efficiently verify the correctness of the quantum result without redoing the computation  themselves.  
In many real-world scenarios, it is not enough for only the original client to verify the result of a delegated quantum computation; instead, the result should be publicly verifiable. For instance, public verifiability becomes essential when (a) the client is offline or resource-constrained, or (b) dispute resolution is needed and third parties (e.g., courts, auditors, or regulatory bodies) must confirm correctness. 
In particular, in commercial settings such as quantum cloud services, public verifiability strengthens accountability and helps prevent spurious disputes: customers and auditors can independently validate reported outputs, and providers are protected against false accusations of misbehavior. More broadly, these deployment requirements (offline or resource-constrained clients, dispute resolution, and external auditability) motivate publicly verifiable CVQC. At the same time, achieving public verifiability is nontrivial because existing CVQC protocols are typically designated-verifier and rely on secret verification keys; Section~\ref{sec::Naive-Approaches} discusses why naive ``publish-the-key'' strategies fail. Publicly verifiable CVQC is therefore particularly relevant in high-stakes domains where outputs affect consequential decisions and post-hoc auditability is required, for example:

\begin{itemize}[leftmargin=4mm]
        \item Drug Discovery: Quantum algorithms like Quantum Approximate Optimization Algorithm  and Variational Quantum Eigensolver  can simulate complex molecular interactions and predict protein folding patterns more efficiently than classical methods. This not only accelerates the drug discovery process but also deepens our understanding of protein misfolding diseases, including Alzheimer's and Parkinson's~\cite{DurantKNDLWYOS24,kandula2023quantum}.

        \item Genomics  and Personalized Medicine:  Quantum computing has the potential to transform genomics and personalized medicine by facilitating the analysis of complex genetic interactions on a scale beyond the reach of classical computers~\cite{chow2024quantum,RasoolARQQA23}.

    
    \item Portfolio Optimization and Risk Modeling: Quantum algorithms can solve complex optimization problems inherent in portfolio management, enabling better asset allocation and risk assessment \cite{yalovetzky2024solving}. For instance, the Harrow-Hassidim-Lloyd  algorithm can be employed for such purposes~\cite{harrow2009quantum}.
    
    

 \item Quantum-Enhanced Classification: Quantum machine learning models such as Quantum Support Vector Machines and Variational Quantum Classifiers can substantially improve data classification~\cite{HavlicekCTHKCG19}. These models have shown promise in fields like fraud detection and medical diagnosis by offering potential computational advantages over classical classifiers~\cite{abs-2308-05237,maouaki2024quantum}.


    \end{itemize}

\subsection{Naive Approaches}\label{sec::Naive-Approaches}

One might wonder why the client cannot use a privately verifiable scheme and publish the secret key in the clear alongside the circuit to enable third-party verification. Publishing the key undermines security: once it becomes public, a malicious prover can forge proofs that pass verification.
Even if the client privately shares the key with a select set of verifiers, verification remains non-public, because only those parties can check the result. Worse yet, if any verifier leaks the secret to the quantum server, the server could generate convincing but invalid proofs, completely breaking soundness.

Alternatively, if the client were allowed to send a second message, it could keep the secret verification key private until the server returns the result and proof, and then publish the key so that anyone can verify. In our setting, this ``publish-later'' strategy is not available: it requires either (i) continued client availability to issue a second release message, or (ii) a trusted escrow that stores the key and releases it at the appropriate time. We assume neither, the client may go offline immediately after its first-round message, and we do not rely on any trusted party for timed key release. Consequently, delayed public verification must embed the key material in the client's initial message in a way that prevents recovery before the prescribed time and lets any party later recover the same key fixed at setup and check its consistency with the public parameters.





\subsection{Technical Overview}


Before presenting the high-level technical ideas behind our approach, we first highlight the core challenges by examining several strawman solutions which may at first glance appear viable but ultimately fail to address our goals.

\subsubsection{Initial Attempts.} 
As previously mentioned, SNARGs for $\BQP$ are the best candidate for solving this problem. However, their construction currently relies on very strong assumptions~\cite{bartusek2022a}. One might consider adopting an approach similar to \cite{chia2020} (in the CRS model) where a fully homomorphic encryption (FHE) of the private verification key is sent by the client to third parties where they will, in turn, run a (publicly executable) FHE evaluation of the verification procedure over the ciphertext to output an encrypted result of the verification. However, the third-party verifier will not be able to decrypt the prover's homomorphically evaluated result, as they do not possess the FHE secret key. 

To avoid dealing with encrypted results, a somewhat related attempt is to use functional or even attribute-based encryption (ABE) as proposed in \cite{parno2012} where a delegated Boolean circuit $C$ is represented as a functional secret key $sk_C$, the input $x$ to the circuit is represented as an attribute on a ciphertext $c_x$ of a randomly encrypted message, and computing the result consists of decrypting $c_x$ using $sk_C$ to get $C(x)$. That said, because $C$ is a quantum circuit, we need ABE for circuits $C$ in $\BQP$, which again is only known to be achievable from iO \cite{bartusek2022a}. Alternatively, we can try to let the functional secret key be $sk_{V_{vk}(.)}$ where $V$ is the \textit{classical} verifier circuit with the private verification key hardcoded. For this, we need some form of public-key function-private ABE, which is known to be restricted to a limited class of functions, e.g., point functions or inner products~\cite{brakerski2018a}.

To circumvent the apparent lack of sufficiently expressive primitives from standard assumptions, we elected to relax the public verifiability requirement to be active only after the computation and its proof are released. This intuition leads us towards using time-based cryptography as a means to delay verifiability for the benefit of basing our scheme on sturdier assumptions.

\subsubsection{Our Approach.} 

At a high level, the core idea behind our protocol is to transform an efficient, privately verifiable CVQC scheme into a publicly verifiable one by encapsulating the private verification key within a time-lock puzzle, a cryptographic time capsule. The verification key is released only after the delegated computation is completed.  The steps are illustrated in Figure~\ref{fig:overview}, which we describe in more detail here. 
%
%

The privately verifiable protocol that we start with relies on the measurement protocol of Mahadev \cite{mahadev2018} which generates a secret trapdoor (of a trapdoor injective claw-free function, TCF) for use during the verification. We require the client, during the setup, to encode the secret trapdoor used for the verification in a time-lock puzzle and publish it as part of the public parameters of the scheme. 

%
The client publishes the public parameters along with the puzzle. The server performs the delegated computation and publishes the result along with a proof asserting the computation has been performed correctly. Later, the server finds the solution to the puzzle and publishes the solution too. The puzzle's difficulty is set such that the server cannot find its solution before performing the computation and generating the proof. Thus, if $C(x)$ takes time $T$ to complete and prove, we set the puzzle to be solved in time $\Delta > \Omega(T)$. Given the computation result, proof, and the trapdoor, any verifier can check the solution's validity.  Since the server has already computed the result and published it, the knowledge of the trapdoor that it finds later will not help it to cheat and create a proof for a false statement. 

\begin{figure}[t]
    \centering
        \centering
        \includegraphics[scale=0.35]{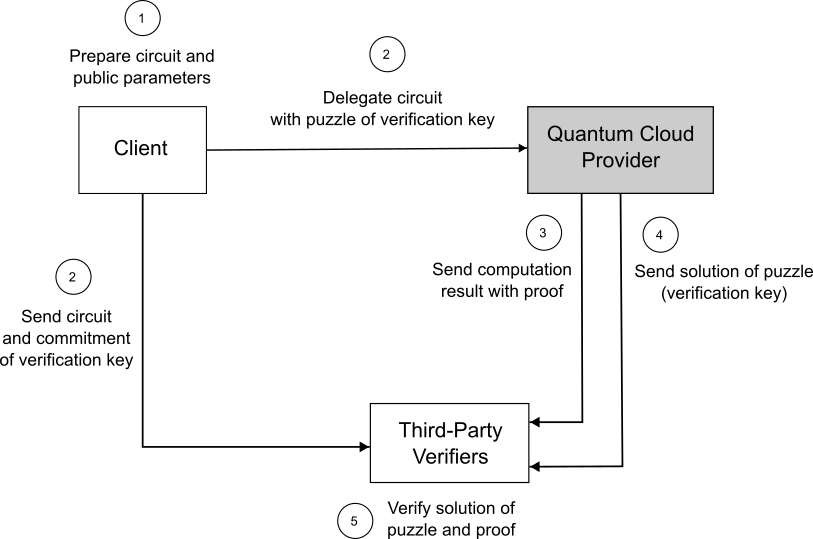}
    %

    \caption{A simplified diagram illustrating the high-level structure of the post-quantum PV-CVQC protocol using time-lock puzzles and commitments.}
    \label{fig:overview}
\end{figure}

Furthermore, the server cannot wait until after solving the puzzle and then publish a forged proof, because the proof must carry a verifiable timestamp $\tau$, and verification rejects unless $\tau$ is smaller than the deadline. The verifiers do not need to monitor the server; they only check the timestamp embedded in the proof against the public deadline. 

However, we still need to make sure the verifiers are given the original trapdoor, otherwise the server might be able to provide an alternative trapdoor for an invalid computation result but still convince the verifiers to accept the proof. 
To satisfy this requirement, we use the following efficient technique, introduced in~\cite{AbadiK21}. When the client generates the time-lock puzzle for the trapdoor, it generates a commitment to this trapdoor and publishes the commitment as part of the public parameters. During the generation of the time-lock puzzle, the client concatenates the plaintext trapdoor with the fixed-length random value used for the commitment and generates the puzzle for this plaintext, i.e., the concatenation. In this setting, when the server solves the puzzle, it parses the solution into two parts: the trapdoor and the randomness (which are the opening of the commitment) and publishes them. Using the opening, a verifier initially checks whether they match the public commitment and if the verification passes, then the verifier uses the trapdoor to check the correctness of the computation's result. 
To ensure that the time-lock puzzle cannot be solved by the quantum prover faster than intended, it must be based on a post-quantum hardness assumption. Since TCF was shown to be realized under the LWE or Ring-LWE assumptions \cite{brakerski2018, brakerski2020a}, which is secure against poly-time quantum adversaries, we can rely on the same assumptions to realize a time-lock puzzle primitive as recently shown in~\cite{lai2023, agrawalr2024}. 

\noindent To summarize, informally, the protocol steps are as follows:
\begin{enumerate}
    \item The client, who wants to delegate a circuit $C$ on input $x$, generates the public $pk$ and private $sk$ parameters of the underlying privately verifiable CVQC scheme.
    \item It publishes a time-lock puzzle $o$ and a commitment $d$ of the private trapdoor alongside the circuit to be delegated.
    \item The server runs the CVQC prover's algorithm on the delegated computation $C(x)$. It publishes the computation output along with the associated proof of correctness. Concurrently, the server may start solving the puzzle $o$.
    \item Once the puzzle solution, representing the now-public trapdoor, is recovered, the server publishes it to all verifiers.
    \item The verifier, now having access to the CVQC trapdoor verification key, can check the validity of the puzzle solution against the commitment $d$, and if it passes, initiates the CVQC verification algorithm to either accept or reject the proof.
\end{enumerate}
We elaborate on each of these steps and the implementation of each algorithm more formally in Section~\ref{sec:pvCVQC}.

%% file: sections/02_related.tex
\section{Related Work}
\label{sec:relatedWork}

This section covers previous work in verifiable quantum computation and public delegation as well as techniques from timed cryptography that are adjacent to our work. Although there have been several works addressing the verifiable quantum computation problem, we will focus mainly on the setting of a single classical verifier and single quantum prover, the most notable of which are listed in Table~\ref{table:relatedwork}.

\input{sections/comparison-table}

\subsection{Verification of Quantum Computation}
Approaches used for quantum computation verification can be categorised according to the model of computation that parties are assumed to possess within a particular setting and whether certain computational assumptions are made. If the verifier is imbued with some lightweight quantum capabilities (e.g., measurement), then various solutions exist in this setting \cite{barz2013,fitzsimons2017,morimae2016,morimae2018,takeuchi2018,kashefi2024}. There are also methods in which the verifier is kept classical, but the protocol requires two or more entangled provers \cite{ji2016,grilo2019,coladangelo2019}. While these protocols can be shown to have information-theoretic security, having multiple entangled and non-communicating provers is often practically prohibitive.

Borrowing from ideas in computational security, a recent promising line of work, initiated by Mahadev \cite{mahadev2018}, deviates from the above approaches by making additional computational assumptions to obtain verifiable delegation in a setting where the verifier is an efficient classical machine interacting with a single (computationally bounded) prover. This was the first such interactive protocol in this setting, verifying the result of a quantum computation using four rounds of communication. This protocol was proven to be sound under the learning with errors (LWE) assumption \cite{regev2005}, which is assumed to be intractable for polynomial-time quantum machines. Subsequent works \cite{chia2020, alagic2020} reduced the number of rounds to 2 (or one round with setup) in the quantum random oracle model, and in fact, we use these as a basis for our protocol.


Gheorghiu and Vidick \cite{gheorghiu2019} propose a composable protocol that allows a verifier to  delegate to a prover the preparation of certain single-qubit quantum states. The composability feature of the protocol allows it to be securely combined with other protocols. The protocol, in a preprocessing phase, requires several rounds of communication that depend on the circuit size. It also relies on LWE. Chung et al. \cite{ChungLLW22} propose a generic compiler that transforms any CVQC protocol to a blind one while preserving the original protocol's round complexity. It relies on quantum fully homomorphic encryption schemes and quantum-secure LWE.   

Delegation protocols can also have the property of blindness where the prover is guaranteed to learn nothing about the circuit being delegated (apart from its size). However, impossibility results have been established for \textit{information-theoretically} secure blind delegation of quantum computation to classical verifiers~\cite{aaronson2019}.

All of the aforementioned protocols are in the designated verifier model where the verifier needs to possess some secret information for verification; if the prover gains access to this secret then it can convince a verifier about an invalid proof. As a result, they are not publicly verifiable. 




\subsection{Publicly Verifiable Delegation}

Publicly verifiable delegation has long been studied for the setting where the verifier and prover are both classical machines. This is reflected in the long line of work that proposes various schemes under standard \cite{kalai2019} and non-standard or idealized (e.g., Random Oracle Model) assumptions \cite{bitansky2013, bitansky2017}. Commonly used methods for achieving publicly verifiable proof systems involve the application of the Fiat-Shamir heuristic \cite{fiat1987} to convert a public-coin interactive proof system (e.g. sigma protocols) into a non-interactive one \cite{lindell2015,canetti2019,bunz2018} or rely on trusted setup assumptions \cite{groth2010a,bitansky2013b,peikert2019} such as common reference strings (CRS). Such techniques generally give rise to non-interactive zero-knowledge proofs (NIZKs) and succinct non-interactive arguments of knowledge (SNARKs) with the public verifiability property. 

In the quantum setting, several constructions of NIZKs for QMA have been proposed under different settings and/or assumptions. The works of \cite{coladangelo2020,morimae2022} construct NIZK for QMA with quantum pre-processing in the designated-verifier model where the proof is classically verifiable but the verifier needs to generate a quantum message as its first message. In addition, several constructions were proposed under the malicious designated verifier (MDV) model \cite{shmueli2021,bartusek2021a} where the trusted setup generates a \textit{reusable} common random string but any verifier, given the CRS, needs to generate its own classical public-secret key pair where the secret key is used for its own private (designated) verification. All of the above protocols require some form of secret parameters for verification and are therefore not publicly verifiable.   


Previous work on publicly verifiable quantum computation relied on the verifier having some form of quantum capability \cite{honda2016, liu2023}. In the CVQC setting, a recent line of work~\cite{bartusek2022a, bartusek2023a, bartusek2026} focuses on achieving public verifiability based on (post-quantum) indistinguishability obfuscation (iO) for classical circuits or under the classical oracle model, which idealizes obfuscation for classical circuits. Our work aims to eliminate the (strong) iO requirement. While follow-ups to this work \cite{metger2024, gunn2025} did achieve succinct arguments for \QMA by forgoing iO and relying only on (post-quantum) LWE only, their schemes are interactive with constant round complexity, while we wish to keep the communication non-interactive.

\subsection{Time-Based Cryptography}

\subsubsection{Time-Lock Puzzle.} 

Timothy May \cite{TimothyMay1993} initially put forth the idea of sending information to the future. A basic property of a time-lock scheme is that generating a puzzle takes less time than solving it. 
The scheme that Timothy May proposed used a trusted agent that releases a secret on time for a puzzle to be solved. Since relying on a trusted agent can be a strong assumption, Rivest \textit{et al.} \cite{Rivest:1996:TPT:888615} proposed an RSA-based TLP. This scheme does not require a trusted agent, relies on sequential squaring, and is secure against a receiver who may have many computational resources that run in parallel with the goal of finding the solution faster.   
Since the introduction of the RSA-based TLP, different variants of it have been proposed, including post-quantum  TLPs \cite{lai2023,agrawalr2024} and homomorphic TLPs, e.g., in  \cite{MalavoltaT19}. 
Malavolta and Thyagarajan et al.  \cite{MalavoltaT19} proposed the notion of homomorphic TLPs, which allows an arbitrary function to run over puzzles of different clients before they are solved. The schemes use the RSA-based TLP  and fully homomorphic encryption. 
Furthermore, to achieve efficiency, partially homomorphic TLPs have also been proposed, including those that support homomorphic linear combinations  of puzzles \cite{MalavoltaT19,liu2022towards}.


\subsubsection{Verifiable Delay Function (VDF)} A VDF enables a prover to provide a publicly verifiable proof stating that it has performed a pre-determined number of sequential computations \cite{BonehBBF18,Wesolowski19,BonehBF18,Pietrzak19a}. VDF was first formalized by Boneh \textit{et al}  \cite{BonehBBF18}. They proposed several VDF constructions based on SNARKs along with either  incrementally verifiable computation or injective polynomials, or based on time-lock puzzles, where the SNARKs-based approaches require a trusted setup.  Later,  Wesolowski \cite{Wesolowski19} and Pietrzak \cite{Pietrzak19a} concurrently improved the previous VDFs  from different perspectives and proposed schemes based on sequential squaring. They also support efficient verification. Most VDFs have been built upon TLPs. Nevertheless, the converse is not necessarily the case because VDFs are not suitable to encapsulate an arbitrary private message since they take a public message as input, whereas TLPs have been designed to conceal a private input message. 
%


\subsubsection{Timed Commitment.} A timed commitment is a commitment scheme that includes an optional forced opening phase, enabling the receiver to recover the committed value without the committer’s cooperation—provided the receiver invests sufficient computational effort and time~\cite{BonehN00,abusalah2024,katz2020}. This property is particularly useful in scenarios where the sender refuses to voluntarily open the commitment. In our solution, a timed commitment can be employed as a black-box component to unify the time-lock puzzle and commitment schemes, albeit at the expense of efficiency. This is because, in our scheme, the client is assumed to be honest and therefore does not need to prove (e.g., via zero-knowledge proofs or related techniques) that the puzzles have been generated correctly, an assurance that is typically required in timed commitment schemes. Therefore, we adopt an efficient technique proposed in~\cite{AbadiK21}, which combines a time-lock puzzle with a commitment scheme, enabling the server to efficiently prove that it has recovered the correct solution originally embedded by the client within the puzzle.

\subsubsection{Timed Signatures and Proofs.} Researchers have proposed the timed version of digital signatures \cite{AbadiCKZ20}. They formally defined such a notion in the universally composable (UC) framework \cite{Canetti01} and proposed a protocol to realize this concept by mainly relying on a blockchain and standard digital signature scheme. Arun et al. \cite{arun2022} proposed the notion of short-lived zero-knowledge proofs and signatures, which consider the case that after a certain time period the proofs or signatures will be invalid/forgeable. The main idea behind the aforementioned notion is to allow any party to forge any proof by executing
a large sequential computation; the proposed schemes primarily rely on VDFs. This differs from our setting where instead of eventual deniability, we are concerned with eventual public verifiability. To enable a smart contract to efficiently verify a proof in a  ``proofs of data retrievability''  system, researchers combined a message authentication code (MAC) with an RSA-based TLP. The MAC verification key is embedded in the TLP and provided to the smart contract after the proof is registered~\cite{AbadiK21}.


%% file: sections/comparison-table.tex
\begin{table*}[!t]
\centering
\small
\scalebox{1}{
\begin{tabularx}{\textwidth}{ 
  | >{\centering\arraybackslash}p{4.3cm}
  | >{\centering\arraybackslash}p{1.7cm} 
  | >{\centering\arraybackslash}p{1.5cm}
  | >{\centering\arraybackslash}X 
  | >{\centering\arraybackslash}X 
  | >{\centering\arraybackslash}p{2cm}
  | >{\centering\arraybackslash}p{3cm}
  | }
 \hline
 {\textbf{Protocol}} 
 & {\textbf{Round Complexity}}
 & \textbf{Succinct}  
 & \textbf{Blind} 
 & \textbf{Publicly Verifiable} 
 & \textbf{Assumptions} 
 & \textbf{Model} 
 \\ \hline
  Mahadev \cite{mahadev2018} & 4 & No & No & No & LWE & Plain \\ \hline
  Gheorghiu and Vidick \cite{gheorghiu2019} & $O(|C|)$  & No & Yes & No & LWE & Plain  \\ \hline
  Chia et al. \cite{chia2020} & 2 & Yes & No & No & LWE + iO & QROM+CRS  \\ \hline
  Alagic et al. \cite{alagic2020} & 1 (w/setup) & No & No & No & LWE & QROM \\ \hline  
  Zhang \cite{zhang2022} & $O(\secp)$ & Yes ($O(|C|)$) & No & No & LWE & QROM \\ \hline
  Bartusek and Malavolta \cite{bartusek2022a} & 1 & Yes & No & Yes & LWE + iO & QROM + CRS \\ \hline
  Metger et al. \cite{metger2024} and Gunn et al. \cite{gunn2025} & $O(1)$ & Yes & No & No & LWE & Plain \\ \hline
  Bartusek et al. \cite{bartusek2023a,bartusek2026} & 1 & Yes & Yes & Yes & LWE & Classical Oracle Model \\ \hline
 This work & 1 (w/setup) & No & No & Yes & SIS/LWE variants & QROM + CRS \\ \hline 
\end{tabularx}
}
\
\caption{The existing CVQC schemes in the single verifier, single prover setting and their properties. Succinctness refers to the property that verification time is $\poly(\secp, \log d(C))$ where $\secp$ is the security parameter and $d(C)$ is the depth of the delegated circuit $C$. Blindness refers to keeping the computation hidden from the prover. All computational assumptions are post-quantum.
}
\label{table:relatedwork}
\vspace{-3mm}
\end{table*}

%% file: sections/03_prelims.tex
\section{Preliminaries}
\label{sec:preliminaries}

In this section, we present the background and necessary material for the models, primitives, and protocols used to build our scheme.

\subsection{Notation} 

We say that a function $f : \N \rightarrow [0,1]$ is negligible if, for any polynomial $p(.)$ and sufficiently large $n \in \N$, we have $f(n) < 1/p(n)$. We denote $\cF(X,Y)$ to be the set of all functions with domain $X$ and range $Y$. Given two interactive algorithms $P$ and $V$, we denote $y \xleftarrow{\$} \angles{P(w),V(z)}(x)$ as the random variable representing the interaction between $P$ (with private input $w$) and $V$ (with private input $z$) on common input $x$ where $y$ is $V$'s local output. 

We say that a classical algorithm $A(x)$ is efficient or probabilistic polynomial time (PPT) if it runs in time $\poly(|x|)$. Similarly, a quantum algorithm $Q(\ket{x})$ is said to be a quantum polynomial time (QPT) machine if it runs in time $\poly(|x|)$ where $\ket{x}$ is a state in some Hilbert space $\cH$. For any given circuit $C$, we denote $|C|$ to be the size of the circuit and $depth(C)$ to be its depth. We denote $\BQP$ to be the set of languages decidable by a QPT machine and $\QMA$ to be the set of all languages $L$ where $L \in \QMA$ if, for every $x \in L$, there exists a polynomial-sized quantum state (a witness) that makes a QPT machine accept $x$ with high probability (and rejects otherwise). For any two distributions $X$ and $Y$, we denote $X \compind Y$ whenever $X$ and $Y$ are computationally indistinguishable and $X \statind Y$ to mean that $X$ and $Y$ are statistically close. Throughout our work, we refer to the Quantum Random Oracle Model~\cite{boneh2010} (QROM) as an idealized model that allows adversaries to have black-box access to a quantum analogue of the classical random oracle that can input/output quantum states.

\subsection{Time-Lock Puzzles}\label{sec::Time-lock-Encryption} 

In this section, we restate the  definition of a time-lock puzzle (TLP)~\cite{Rivest:1996:TPT:888615}. 

\begin{definition}
\label{Def::Time-lock-Puzzle} 
A TLP scheme consists of three algorithms: $\TLP = (\TLPSetup,$ $ \TLPGen, \TLPSolve)$ defined as follows:
\begin{itemize}[leftmargin=4.2mm]
\item $\TLPSetup(1^{\secp},\Delta)\rightarrow (tpk,tsk)$. A probabilistic algorithm that takes as input a security parameter, $1^{\secp}$, and time parameter $\Delta$ that specifies how long a message must remain hidden in seconds. 
It outputs a pair $(tpk,tsk)$ that contains the scheme's public and private parameters, respectively.

\item $\TLPGen(m, tpk, tsk)\rightarrow {o}$. A probabilistic algorithm that takes as input a solution $m$ and $(tpk,tsk)$. It outputs a puzzle $o$.

\item $\TLPSolve(tpk, {o})\rightarrow s$. A deterministic algorithm that takes as input  $tpk$ and $ {o}$. It outputs a solution $s$.

\end{itemize}

The following properties must also be satisfied:
\begin{itemize}
\item Completeness. For any message $m$ it always holds that: $$\TLPSolve(tpk, \TLPGen(m,tpk,tsk))=m$$

\item Efficiency. The run-time of  $\TLPSolve(pk, {o})$ is upper-bounded by $poly(\Delta,\secp)$.
\end{itemize}

\end{definition}

In certain schemes, such as \cite{agrawalr2024}, the secret key $tsk$ may not be needed. The security of a TLP requires that the puzzle's solution remains confidential against all adversaries running in parallel within the time period, $\Delta$. It also requires that an adversary cannot extract a solution in time $\delta(\Delta)<\Delta$, using a polynomial number of processors $poly'(\Delta)$ processors that run in parallel and after a large amount of pre-computation. 

\begin{definition}
\label{def:tlpsecurity}
A TLP is secure if for all $\secp$ and $\Delta$, all quantum polynomial time (QPT) adversaries $A:=(A_{   1},A_{2})$ where $A_{1}$ runs in total time $O(poly(\Delta,\secp))$ and $A_{2}$ runs in parallel time $\delta(\Delta)<\Delta$ using at most $p(\Delta) = \poly(\Delta)$ parallel processors, there is a negligible function $\negl(\secp)$, such that: 


$$ \Pr\left[
  \begin{array}{l}
(tpk, tsk) \gets \TLPSetup(1^{\secp},\Delta) \\
(m_{   0},m_{   1},\text{state}) \gets A_{   1}(1^{  \secp},tpk, \Delta)\\
b  \stackrel{  \$}\leftarrow \{0,1\}\\
{o} \gets \TLPGen(m_{   b}, tpk, tsk) \\
\hline
b \gets A_{   2}(tpk,  {o},\text{state}) \\

\end{array}
\right]\leq \frac{1}{2}+\negl(\secp)$$
\end{definition}

In the literature of TLP, the  security definition includes two adversaries $A_1$ and $A_2$ because they have different run-time constraints. In this work, we use the same approach.


\subsubsection{Post Quantum Secure TLPs.}



Lai and Malavolta \cite{lai2023} proposed a candidate post-quantum secure sequential function relying on a lattice-based hash function \cite{mahmoody2011}. Using the proposed sequential squaring, it is possible to easily construct a TLP. However, this TLP will require the puzzle generator to take as many steps as the puzzle solver. Very recently, Shweta \etal \cite{agrawalr2024} proposed a TLP based on lattices that can be instantiated using the SIS-sequential function of \cite{lai2023} (with judicious parameter settings). In this scheme, the puzzle generator can create puzzles more quickly than the solver can solve them. Our scheme can leverage any post-quantum secure TLP in a black-box manner. 

\begin{lemma}[\cite{agrawalr2024}] Assuming the quantum hardness of circular small-secret LWE \cite{hsieh2023}, there exists a time-lock puzzle that is secure against polynomial-time quantum adversaries.
\end{lemma}

\subsection{Commitment Scheme}\label{subsec:commit}

There are two parties involved in a commitment scheme, a \emph{sender} and a \emph{receiver}. The commitment scheme has two phases, \emph{commit} and  \emph{open}. In the \emph{commit} phase, the sender  commits to a message $m$ as $\comcom(m,r)=com$, that involves a secret value,  $r$. In the \emph{open} phase, the sender sends the opening $\hat{m}:=(m,r)$ to the receiver which verifies its correctness: $\comver(com,\hat{m})\stackrel{\scriptscriptstyle ?}=1$ and accepts if the output is $1$. A commitment scheme must satisfy (a) \textit{hiding}, it is infeasible for an adversary (i.e., the receiver) to learn any information about the committed message $m$, until the commitment ${com}$ is opened, and (b) \textit{binding}, it is infeasible for an adversary (i.e., the sender) to open a commitment ${com}$ to different values $\hat{m}':=(m',r')$ than that was used in the commit phase, i.e., infeasible to find  $\hat{m}'$, \textit{s.t.} $\comver({com},\hat{m})=\comver({com},\hat{m}')=1$, where $\hat{m}\neq \hat{m}'$.

Definition~\ref{def:commitment} in
Appendix~\ref{sec::Formal-Definition-of-Commitment Scheme} provides a formal definition of a commitment scheme.

 There exist efficient commitment schemes in the random oracle model using the well-known hash-based scheme such that committing  is: $\mathtt{H}(x||r)={com}$ and $\comver({com},\hat{x})$ requires checking: $\mathtt{H}(x||r) \stackrel{?}={com}$, where $\mathtt{H}:\{0,1\}^{*}\rightarrow \{0,1\}^{ \lambda}$ is a collision-resistant hash function, i.e., the probability to find $x$ and $x'$ such that $\mathtt{H}(x)=\mathtt{H}(x')$ is negligible in the security parameter $\lambda$. There have also been various post-quantum commitment schemes \cite{kawachi2008,XieXW13,0002KPT12,wee2023} by relying on LWE, learning parity with noise (LPN), or short integer solution (SIS) assumptions. 

 \begin{lemma}[\cite{wee2023}] Assuming the quantum hardness of the Short Integer Solutions (SIS) problem, there exists a computationally binding and statistically hiding commitment scheme.
\end{lemma}

\subsection{Quantum Verification Protocols}

We recall the complexity class $\QPIP_{\tau}$  defined in \cite{aharonov2017} as the set of languages that have an interactive proof between a $\BQP$ prover and a hybrid classical-quantum verifier with a $\tau$-qubit register. In our work, we will focus only on fully classical ($\BPP$) verifiers (i.e. $\tau = 0$). 

\begin{definition}
\label{def:qpip}
A $\QPIP_{0}$ protocol  $\Pi = (P,V)$ for a $\BQP$ language $L$ with completeness $c(.)$ and soundness $s(.)$ is an interactive protocol between a QPT prover $P$ and a PPT verifier $V$ denoted as $b \leftarrow \angles{P, V(z_v)}(1^\secp, z)$ on shared input $z \in \{0,1\}^{*}$ where $b \in \{0,1\}$ is the output of $V$ indicating acceptance/rejection and $z_v \in \{0,1\}^{*}$ is the verifier's private input. Also, the following properties hold:
\begin{itemize}
    \item Completeness: For every $\secp \in \N, z \in \L$, and $z_v \in \{0,1\}^*$:
    \begin{align*}
    \Pr\left[
    \angles{P, V(z_v)}(z, 1^\secp) = 1
    \right]
    \geq c(\secp)
    \end{align*} 

\item Soundness: For any arbitrary QPT prover $P^*$, large enough $\secp \in \N$, and all $z \notin \L$  and $z_v \in \{0,1\}^*$:
    \begin{align*}
    \Pr\left[
    \angles{P^*, V(z_v)}(z, 1^\secp) = 1
    \right]
    \leq s(\secp)
    \end{align*} 
\end{itemize}
where $c-s \geq 1/\poly(\secp)$. Without loss of generality, we assume that $s = \negl(\secp)$.
\end{definition}

Given a $\QPIP_0$ protocol, one can use it to classically verify the computation of any $\BQP$ computation that was delegated to an untrusted quantum prover. This task is referred to in the literature as classical verification of quantum computation (CVQC). Simply put, one can define $L = \{(C,x,y) \mid C(x) = y\}$ to be the set of quantum circuit evaluations and their respective outputs. For simplicity, in this work, we assume that $y = 1$ for all circuit-input pairs $(C,x)$ in $L$ (i.e., circuits are Boolean) and restrict our attention to quantum circuits that accept and output classical strings. Since the classical verifier has no (efficient) way to evaluate this quantum circuit, it falls to the prover to provide a classical proof convincing the verifier of the evaluation's correctness. 

In the original Mahadev CVQC protocol~\cite{mahadev2018}, the public parameter $pk$ represents a sequence of $n = O(|C|)$ public keys $pk_1,\ldots,pk_n$ each used to evaluate a  trapdoor claw-free function (TCF). These keys allow the prover to coherently evaluate the TCFs over its computation instance and send a classical string representing a commitment of its results. The verifier then initiates a 2-round challenge-response mechanism to test and verify the prover's claim.


As demonstrated in \cite{alagic2020,chia2020}, the above 4-round commit-and-measure protocol can be transformed into a non-interactive protocol (with setup) via the round-collapsing Fiat-Shamir transform \cite{fiat1987}, which was shown to be sound in the QROM \cite{don2019,liu2019}. However, this transformation still results in a protocol that is privately verifiable. We describe their scheme as follows:

\begin{proto}[Privately Verifiable CVQC \cite{chia2020}]
\label{proto:2round-cvqc}
Let $x \in L$ be the instance of a $\BQP$ language $L$. Assuming the hardness of LWE, there exists a construction of a 2-round privately verifiable CVQC protocol $\Pi = (V_1,P,V_{out})$ between verifier $V = (V_1,V_{out})$ and prover $P$ in the QROM for language $L$ that operates as follows:
\begin{enumerate}[leftmargin=4.8mm]
    \item $V_1(1^\secp, x)$: Given security parameter $1^\secp$ and an instance $x$, generate a public-private key pair $(pk,sk)$, then send $pk$ to the prover.
    \item $P(pk,x)$: Given $x$ and $pk$, generate a classical proof string $\pi$ and send it to $V$.
    \item $V_{out}(pk,sk,\pi)$: Given $pk$, the private verification key $sk$ and proof $\pi$, verifier $V$ either outputs 1 (accept) or 0 (reject).
\end{enumerate}
\end{proto}


Since the above protocol is a 2-round $\QPIP_0$, it has the same completeness and soundness properties formalized in Definition~\ref{def:qpip}. CVQC schemes may also have the added property of being \textit{publicly verifiable} where the correctness of the computation can be verified by any arbitrary third party (and not just by the delegator) using only the transcript of an interaction. 


%% file: sections/04_protocol.tex
\section{Time-Delayed Publicly Verifiable CVQC}
\label{sec:pvCVQC}

We start by defining a publicly verifiable time-delayed CVQC scheme following the syntax of \cite{parno2012,kalai2019} for classical verifiable computation schemes but with a notion that is closer to publicly verifiable non-adaptive preprocessing SNARGs \cite{bitansky2013b}. Without loss of generality, we focus on the verifiability of Boolean circuits $C$ since in order to prove that $C'(x) = y$ for some $y \neq 1$, it suffices to instead prove that $C_y(x) = 1$ where $C_y$ is the circuit that runs $C'(x)$ and outputs 1 if and only if $C'(x) = y$. It is thus natural to generalize this proof system to $\QMA$ languages of the form $L = \{(C,z) \mid \exists y : C(z) = y\}$.

\begin{definition}\label{def::pvtd-CVQC} A publicly verifiable time-delayed CVQC scheme for a family of Boolean quantum circuits $\{\cC_\lambda: \{0,1\}^{n(\lambda)} \to \{0,1\}\}_{n \in \mathbb{N}}$ with a time-bound $T = T(\secp)$ consists of four algorithms $\VC_T = (\VCSetup,\VCProve,\VCReveal,$ $\VCVerify)$ that are defined as follows:
\begin{itemize}[leftmargin=4.5mm]
    \item $\VCSetup(1^\secp,C,x, T)$: A PPT algorithm that takes as input a security parameter $1^\secp$, a circuit $C \in \cC_\lambda$, input $x \in \{0,1\}^{n}$, and time bound $T$. It outputs a common reference string $crs$.
    \item $\VCProve(crs,C,x)$: A QPT algorithm, where given $crs$, a circuit $C \in \cC_\secp$, and input $x \in \{0,1\}^{n}$, outputs a timestamped proof $\pi_\tau$. 
    \item $\VCReveal(crs)$: A QPT algorithm where, given $crs$ outputs a string $y$ that facilitates verification.   
    \item $\VCVerify(crs,C,x,\pi_\tau,y)$: A PPT algorithm where, given $crs$, a circuit $C \in \cC_\secp$, input $x \in \{0,1\}^{\poly(\secp)}$,  a timestamped proof $\pi_\tau$, and $y$, outputs 1 (accept) or 0 (reject).   
\end{itemize}
The following properties must also be satisfied:
\begin{itemize}[label=$\bullet$]
    \item Completeness: For every $\lambda \in \N, T = T(\secp), C \in \cC_\lambda$, $x \in \{0,1\}^{\poly(\secp)}$ such that $C(x) = 1$, the following must hold:
    
        \begin{align*}
        \Pr\left[
        \begin{array}{c}
        \VCVerify(crs,C,\\ x, \pi_\tau,y)=1
        \end{array}
        \middle| 
        \begin{array}{c}
        crs \leftarrow \VCSetup(1^\secp, C,x, T) \\
        \pi_\tau \gets \VCProve(crs, C,x) \\
        y \gets \VCReveal(crs)
        \end{array}
        \right]
        = 1
        \end{align*}
    \item Efficiency: For every $\secp$, all algorithms $\VCSetup$ runs in time $\poly(\lambda, |C|,|x|, $ $T)$, while $\VCReveal$, $\VCProve$ run in time that is $\poly(\secp, |C|,|x|,T)$, and $\VCVerify$ runs in time that is $\poly(\lambda,$ $\log depth(C), |x|,\log T)$.
    \item $\Delta-$Soundness: 
    For every QPT adversary $A$ that runs in time $\delta(\Delta) < \Delta$ for some $\Delta > T^{1+\epsilon} = \poly(T)$ and $\epsilon > 0$, there exists a negligible function $\negl(\cdot)$ such that:
    \begin{equation*}
        \Pr[\textbf{Exp}_A^{VC}[\secp,T] = 1] \leq \negl(\secp)
    \end{equation*}
    where $\textbf{Exp}_A^{VC}[\secp,T]$ is defined as the following experiment.
    
\end{itemize}

\FrameSep2pt
\begin{framed}
\label{exp:cvqc}
\noindent Experiment $\textbf{Exp}_A^{VC}[\secp,T]$:  
\\ \hspace*{\parindent} 
$(C,x) \gets A(1^\secp)$ where $C \in \mathcal{C}_\secp$ and $x\in \{0,1\}^{\poly(\secp)}$
\\ \hspace*{\parindent}
$crs \leftarrow \VCSetup(1^\secp, C,x,T)$
\\ \hspace*{\parindent}
$(\pi_\tau,y) \gets A^\mathcal{B}(crs,C,x)$
\\ \hspace*{\parindent}
$b_1 \gets \VCVerify(crs,C,x,\pi_{\tau}, y)$ 
\\ \hspace*{\parindent}
$(sk,r) \gets \VCReveal(crs)$
\\ \hspace*{\parindent}
$b_2 \gets \VCVerify(crs,C,x,\pi_{\tau}, (sk,r))$
\\ \hspace*{\parindent}
Output $(b_1 \vee b_2) \wedge (C(x)=0 \vee sk\neq y)$
\end{framed}


\end{definition}

The experiment is essentially divided into two verifications that challenge the validity of the adversary's generated proof, one that is conducted before time $\Delta$ (where $sk$ is hidden) and one that is conducted at time $\Delta$ after $sk$ is revealed by $\VCReveal$. The first pre-$\Delta$ verification tests the adversary's proof against its own guess of the secret key $y$. The second verification tests the \textit{same} proof against the real secret key $sk$ to ensure that a blind pre-$\Delta$ forgery of the proof still fails under $sk$. Only proofs that have been generated before $sk$ is revealed are considered valid forgeries\footnote{In general, a privately verifiable scheme, which we make use of in our protocol, does not guarantee the soundness of proofs that have been generated \textit{after} the secret verification key $sk$ is disclosed.} and, in order to enforce this, algorithms are given access to a trusted ``time-stamping'' service (e.g., public blockchain) $\cB(\cdot)$ ~\cite{mahmoody2013a} that can generate a verifiable time-stamp at time $\tau$ for any string. Thus, in order to win, the adversary must generate a false proof before time $\Delta$ that is valid against its own fake key $y$ or against $sk$, when it is revealed.

\subsection{The Proposed Scheme}

Our construction, shown in Figure~\ref{fig:mainProtocol}, makes use of Protocol~\ref{proto:2round-cvqc}, the 2-round CVQC protocol, along with a time-lock puzzle scheme to achieve public verifiability, and a commitment scheme to ensure the outsourced puzzle containing the verification key was correctly revealed. 
Figure~\ref{fig:overview-detailed} outlines the different phases of the protocol.

\begin{figure*}[!ht]
\centering
\begin{minipage}{\textwidth}
\begin{framed}
\noindent
\begin{proto}[Time-Delayed Publicly Verifiable CVQC]
\label{proto:pv-CVQC}
    Let $\Pi = (V_1,P,V_{out})$ be a 2-round privately verifiable CVQC in the QROM, let $(\comcom, \comver)$ be a commitment scheme, and let $\TLP = (\TLPSetup,\TLPGen, $ $ \TLPSolve)$ be a time-lock puzzle scheme. All algorithms have access to a timestamping service $\mathcal{B}$ (e.g., public blockchain); $\mathcal{B}$ does not escrow secrets. Verifiers may be offline at proof-generation time; they only need access to $\mathcal{B}$’s public record to validate the embedded timestamp later during verification. 
    The construction of the time-delayed publicly verifiable CVQC $\VC_{T}  = (\VCSetup, \VCProve, \VCReveal, \VCVerify)$ for a family of quantum circuits $\cC_\lambda$ proceeds as follows:

\begin{itemize}
    \item \textbf{Setup Phase:} $\VCSetup(1^\secp,C,x,T)$\\ Given security parameter $1^\secp$, circuit $C \in \cC_\lambda$, input $x \in \{0,1\}^{\poly(\secp)}$, and time-bound $T = T(\lambda)$: 
    \begin{enumerate}
        \item Set $\Delta > T^{1+\epsilon}$ for some $\epsilon > 0$
        \item Generate a public-private key pair for the time-lock puzzle $(tpk, tsk) \leftarrow \TLPSetup(1^\secp, \Delta)$
        \item Run $(pk,sk) \gets V_1(1^\secp,(C,x))$
        \item Let $d \gets \comcom(sk,r)$ where $r \xleftarrow{\$} \{0,1\}^\secp$
        \item Generate puzzle $o \gets \TLPGen((sk,r), tpk,tsk)$
        \item Output $crs = (tpk,pk,o,d,\Delta)$. 
    \end{enumerate}
    \item \textbf{Prove Phase:} $\VCProve(crs, C,x)$ \\ 
    Given $crs$, circuit $C \in \cC_\lambda$, and input $x \in \{0,1\}^{\poly(\secp)}$: 
    \begin{enumerate}
        \item Extract $pk$ from $crs$
        \item Run $\pi \gets P(pk,(C,x))$
        \item Timestamp proof: $\pi_\tau \gets \mathcal{B}(\pi)$
        \item Return $\pi_\tau$
    \end{enumerate}
    \item \textbf{Reveal Phase:} $\VCReveal(crs)$ \\ 
    Given $crs$: 
    \begin{enumerate}
        \item Extract $tpk,o$ from $crs$
        \item Run $(sk',r') \gets \TLPSolve(tpk,o)$
        \item Return $y = (sk',r')$
    \end{enumerate}
    \item \textbf{Verify Phase:} $\VCVerify(crs,C,x,\pi_\tau,y)$:\\ Given $crs$, circuit $C \in \cC_\lambda$, input $x \in \{0,1\}^{\poly(\secp)}$, decommitment $y$, and timestamped proof string $\pi_\tau$:
    \begin{enumerate}
        \item If $\tau > \Delta$, output 0
        \item Extract $(pk,d)$ from $crs$
        \item Parse $y = (sk',r')$
        \item Let $b \gets \comver(d,(sk',r'))$. If $b = 0$, output 0
        \item Output $V_{out}(pk,sk',\pi_\tau)$
    \end{enumerate}
\end{itemize}
\end{proto}
\end{framed}
\end{minipage}
\caption{The Time-Delayed PV-CVQC Protocol in the QROM}
\label{fig:mainProtocol}
\end{figure*}


\input{sections/figure-detailed-workflow}

The protocol starts by running the setup phase $\VCSetup$, which would be executed by an honest client who wants to outsource the computation of some circuit $C$ on input $x$. The time bound $T = T(\cC_\lambda)$ is a function of the circuit family to which the computation belongs and can be set to be an upper bound on the depth of any $C \in \cC_\lambda$. The algorithm $\VCSetup$ will then set the time-lock puzzle difficulty parameter $\Delta$ to be any value greater than $T + T_{prove}$ where $T_{prove}$ is the running time of the prover's $P$ algorithm. This ensures that the computation is guaranteed to be completed before the puzzle is solved. The verification algorithm of the privately verifiable algorithm $V_1$ is then run to obtain the instance-specific public-private parameter $(pk,sk)$ pair of $\Pi$ where the instance is a circuit-input pair $(C,x)$. Finally, the puzzle and commitment of $sk$ are  generated and the CRS is constructed to contain the public parameter $tpk$ for solving the puzzle, the public key $pk$ for running the prover of $\Pi$, the puzzle $o$ and commitment $d$ of $sk$. Recall that the CRS is not sampled uniformly at random but is deliberately structured, which prevents us from simply replacing it with the output of a random oracle. 

The $\VCProve$ algorithm,  typically executed by the server, is used to run the computation and proof generation procedure of $\Pi$ using the public key $pk$ and the statement $(C,x)$. Concurrently, $\VCReveal$ is also executed, which starts solving the puzzle provided in the CRS. Since the puzzle has difficulty parameter $\Delta > T^{1+\epsilon}$, it is expected that $\VCProve$ would complete and output the proof $\pi$ before $\VCReveal$ outputs the solution (the verification key).

Finally, when it comes to the verifier, it will run $\VCVerify$, which will execute the verification algorithm $V_{out}$ of $\Pi$ and output the result since it now has access to the key $sk$ and the proof $\pi$.


\begin{theorem} 
\label{thm:pv-cvqc}
If $\comcom$ is a secure post-quantum commitment scheme, $\TLP$ is a secure post-quantum time-lock puzzle, and $\Pi$ is a 2-round privately verifiable CVQC in the QROM, then Protocol~\ref{proto:pv-CVQC} is a secure time-delayed publicly verifiable non-interactive CVQC protocol for a circuit class $\cC$ under the quantum random oracle with CRS model, w.r.t. Definition~\ref{def::pvtd-CVQC}.
\end{theorem}

To prove the theorem we need to show that it satisfies the correctness, efficiency, and soundness properties of Definition~\ref{sec:pvCVQC}. Completeness follows from the underlying 2-round CVQC protocol, the time-lock puzzle, and the commitment scheme. Specifically, for any $C \in \cC_\lambda$, $x \in \{0,1\}^{n}$ and $T = T(\lambda)$ and any honestly generated $crs = (tpk,pk,o,d, \Delta)$:
$$
\VCVerify(crs,C,x,\pi,y) = V_{out}(pk,sk,\pi) = 1    
$$
where $\pi \gets P(pk,(C,x))$, $d \gets \comcom(sk,r)$, and $y = (sk,r) \gets \TLPSolve(tpk,o)$.

Regarding the efficiency requirement, it is clear that $\VCSetup$ will run in time $\poly(\secp,|C|,|x|,T)$ due to the size of $pk$. Furthermore, $\VCProve$ runs in time $\poly(\secp,|C|,|x|,T)$ as it will need to compute the circuit and the proof, and $\VCReveal$ runs in time $\poly(\secp,T)$ due to $\TLPSolve$, which will solve the puzzle in time $\Delta = \poly(T)$. 



\begin{remark}[Instance-Independent Variant]
In our construction, the CRS is circuit-dependent since the underlying 2-round CVQC requires this circuit to generate the verification key. Alternatively, Alagic et al. \cite{alagic2020} show a variant of a 2-round privately verifiable CVQC that is instance-independent at the cost of increasing the number of parallel executions by a constant. This allows us to make $\VCSetup$ generate a CRS that is independent of the circuit-input pair $(C,x)$. 
\end{remark}

\subsection{Proof of Security}
We prove that the protocol has $\Delta-$soundness regarding Definition~\ref{proto:pv-CVQC}.
\begin{proof}
We show through a sequence of hybrids that the security of our scheme can be reduced to the soundness of the underlying 2-round CVQC, the security of the TLP, and the commitment scheme.
\begin{itemize}
    \item $H_0$: This is the original experiment as stated in Definition~\ref{exp:cvqc} where the adversary's role as a malicious prover is to find a convincing false proof that passes verification.
    \item $H_1$: This is the same as $H_0$ except that Step 5 of the Verify phase is replaced with the following:
    \begin{equation*}
        \text{Output} \; V_{out}(pk,sk,\pi)
    \end{equation*}
    In other words, the verification procedure of the underlying CVQC will use the original $sk$ instead of $sk'$.
    \item $H_2$: This is the same as $H_1$ except that, in Step 5 of the Setup phase, the puzzle is generated using a random string pair $sk^* \xleftarrow{\$} \{0,1\}^{n}$, where $n = |sk|+|r|$, instead of $(sk,r)$. That is, the puzzle is generated as follows:
    \begin{equation*}
        o^* \gets \TLPGen(tpk,tsk,sk^*)
    \end{equation*}
    \item $H_3$: The is the same as $H_2$ except that, in Step 4 of the Setup phase, the challenger commits to the same random string $sk \xleftarrow{\$} \{0,1\}^{n}$ used in the puzzle, instead of $sk$. That is, the commitment is generated as follows:
    \begin{equation*}
        d^* \gets \Commit(sk^*) 
    \end{equation*}
\end{itemize}
Let $\text{Adv}^{VC}_{H_i}(A) = |\Pr\limits_{H_i}[\textbf{Exp}_A^{VC}[\secp,T] = 1]$ denote the adversary's success probability in the soundness game of Definition~\ref{def::pvtd-CVQC}. The following sequence of claims will show that $\text{Adv}^{VC}_{H_0}(A) \leq \negl(\secp)$.

\begin{claim}  
\label{claim:h0}
If the commitment scheme is computationally binding against poly-time quantum adversaries, then:
$$|\text{Adv}^{VC}_{H_0}(A) - \text{Adv}^{VC}_{H_1}(A)| \leq \negl(\secp)$$
\end{claim}
\begin{proof}
Let $A$ be a QPT algorithm that succeeds in distinguishing between $H_0$ and $H_1$ with non-negligible advantage. Then there exists an efficient attacker $B$ that acts as the challenger for $A$ and breaks the binding property of the commitment scheme (as stated in Definition~\ref{def:commitment}) as follows. The attacker $B$ would simulate the original game for $A$ exactly up until Step 5 of the Verify phase. Here, $H_0$ and $H_1$ will differ only if $sk' \neq sk$ but $\comver(d,(sk',r')) = 1$, which would imply breaking the binding property of the commitment scheme. Thus, as long as the commitment scheme is computationally binding the adversary's advantage of distinguishing between $H_0$ and $H_1$ is negligible. 
\renewcommand{\qedsymbol}{} 
\qedhere$\blacksquare$
\end{proof}

\begin{claim}  
\label{claim:h1}
Assuming that the TLP is secure against poly-time quantum adversaries, it holds then $|\text{Adv}^{VC}_{H_1}(A) - \text{Adv}^{VC}_{H_2}(A)| \leq \negl(\secp)$.
\end{claim}
\begin{proof}
Let $A$ be a QPT adversary that succeeds in distinguishing between $H_1$ and $H_2$ with non-negligible advantage. We  show that there exists an efficient attacker $B = (B_1,B_2)$ that acts as the challenger for $A$ and breaks the security of the underlying TLP scheme according to Definition~\ref{def:tlpsecurity} as follows. The attacker $B_1$ starts by running each step of $\VCSetup(1^\secp,(C,x))$, except in Step 2 it will first receive $tpk$ from the TLP challenger, and in Step 5 it will ask for the challenge puzzle by submitting $(m_0,m_1)$ as its chosen messages where $m_0 = (sk,r)$ contains the secret key from Step 3 and $m_1 = sk^*$ is some uniformly random string of size $|sk|+|r|$. Once it receives back the challenge puzzle $o_b \gets \TLPGen(m_b,tpk,tsk)$, it runs $A$ on $crs_b = (tpk,pk,o_b,d)$, which  generates a proof $\pi$. 

Similarly, $B_2$ will execute $A(crs_b,C,x)$. We note, since the running time of $A$ is $\delta(\Delta) < \Delta$, $B_2$ will also have the same running time. Furthermore, if the TLP challenger chooses to generate a puzzle for $m_0 = (sk,r)$ then $crs_b$ will have a puzzle of $sk$ and the view of $A$ will be that in $H_1$. On the other hand, if the TLP challenger chose to generate a puzzle for $m_1 = sk^*$ then the $crs_b$ will have a puzzle of $sk^*$ and the view of $A$ will be that in $H_2$. The rest of the game remains unchanged between the two hybrids. Since $B$ simulates $A$'s challenge almost perfectly (with some negligible loss when $sk^* = (sk,r)$), the advantage of $B$ is equal to that of $A$. Thus, as long as the TLP is secure, the adversary's advantage of distinguishing between $H_1$ and $H_2$ is negligible.
\renewcommand{\qedsymbol}{} 
\qedhere$\blacksquare$
\end{proof}

\begin{claim} 
\label{claim:h2}
Assuming that the commitment scheme is statistically hiding against quantum adversaries, it holds that $H_2 \statind H_3$.
\end{claim}
\begin{proof}
Let $A$ be a QPT adversary that can distinguish between $H_2$ and $H_3$ with a non-negligible advantage. Then we can construct an adversary $B$ that acts as the challenger for $A$ and uses it to break the hiding property of the underlying commitment scheme. The attacker $B$ starts by running each step in the Setup exactly as in $H_2$ except now in Step 4 it will first submit $(m_0,m_1)$ to the challenger in the hiding game of the commitment scheme where $m_0 = (sk,r)$ and $m_1 = sk^*$ for some random $sk^* \xleftarrow{\$} \{0,1\}^{|sk|+|r|}$. Once it gets back the challenge commitment $d_b \gets \comcom(m_b)$, it runs $A$ on $crs_b = (tpk,pk,o^*,d_b, \Delta)$ and $(C,x)$. The rest of the steps remain the same. If the challenger chose to create a commitment for $m_0 = (sk,r)$ then the view of $A$ will be that in $H_2$. If a commitment of $m_1 = sk^*$ was generated instead then the view of $A$ will be that in $H_3$. Since $B$ simulates $A$'s challenge almost perfectly (with some negligible loss when $sk^* = (sk,r)$), the advantage of $B$ in breaking the commitment scheme is equal to that of $A$.
\renewcommand{\qedsymbol}{} 
\qedhere$\blacksquare$
\end{proof}


\begin{claim} 
\label{claim:2cvqc-reduction}
If the 2-round privately verifiable CVQC scheme is computationally sound against all malicious QPT provers, $\comcom$ is computationally binding/hiding, and the TLP is secure, then the advantage of $A$ in $H_3$ is negligible.
\end{claim}
\begin{proof}
Let $A$ be a QPT adversary that can break soundness in $H_3$. We then show how to construct an adversary $B$ that will act as $A$'s challenger and use it to break the soundness of the underlying privately verifiable CVQC scheme. $B$ starts by requesting $pk$ from its challenger ($sk$ will be kept secret from $B$), creates a random puzzle $o^*$ and commitment $d^*$ of a random string, then sends $crs = (tpk,pk,o^*,d^*,\Delta)$ to $A$. Once $A$ sends its proof $\pi_\tau$ and decommitment $y$, $B$ will run $b_1 \gets \VCVerify(crs,C,x,\pi_\tau,y)$ and, after it reveals $(sk^*,r)$, also runs $b_2 \gets \VCVerify(crs,C,x,\pi_\tau,(sk^*,r))$. Since $B$ exactly simulates the view of $A$ in this game, we have that $\Pr[\text{B wins}] = \text{Adv}^{VC}_{H_3}(A) \leq \Pr[b_1 \vee b_2 | C(x) = 0 \vee y \neq sk^*] \leq \Pr[b_1 | C(x)=0 \vee y \neq sk^*] + \Pr[b_2|C(x)=0 \vee y \neq sk^*]$.

We note that if $b_1 = 1$ while $C(x) = 0$ and $y \neq sk^*$ then this implies that $A$ found $sk'\neq sk^*$ which passes $\comver(d^*,sk')$ or directly found $sk^*$ from the commitment $d^*$ or puzzle $o^*$ (before time $\Delta$). However, by the security of the underlying primitives, this only happens with negligible probability. Thus $\Pr[b_1 = 1|C(x)=0 \vee y \neq sk^*] \leq \negl(\secp)$. Observe also that, during the second verification phase (post key release), if $b_2 = 1$ then it must be that $y = sk^*$ yet since $C(x) = 0$ this implies that $\pi_\tau$ has passed as a valid forged proof against the real secret key $sk$, violating the soundness of the underlying CVQC protocol.
%
%
\renewcommand{\qedsymbol}{} 
\qedhere$\blacksquare$
\end{proof}
Putting it all together to prove Theorem~\ref{thm:pv-cvqc}, we get that $H_0 \compind H_1 \compind H_2 \statind H_3$, $\text{Adv}^{H_3}_A \leq \negl(\secp)$. Hence, the advantage of adversary $A$ winning in $H_0$ is negligible.
\end{proof}

\input{sections/table-results}

%% file: sections/figure-detailed-workflow.tex
\begin{figure}[H]
    \centering
    
    %
%
        \includegraphics[scale=0.25]{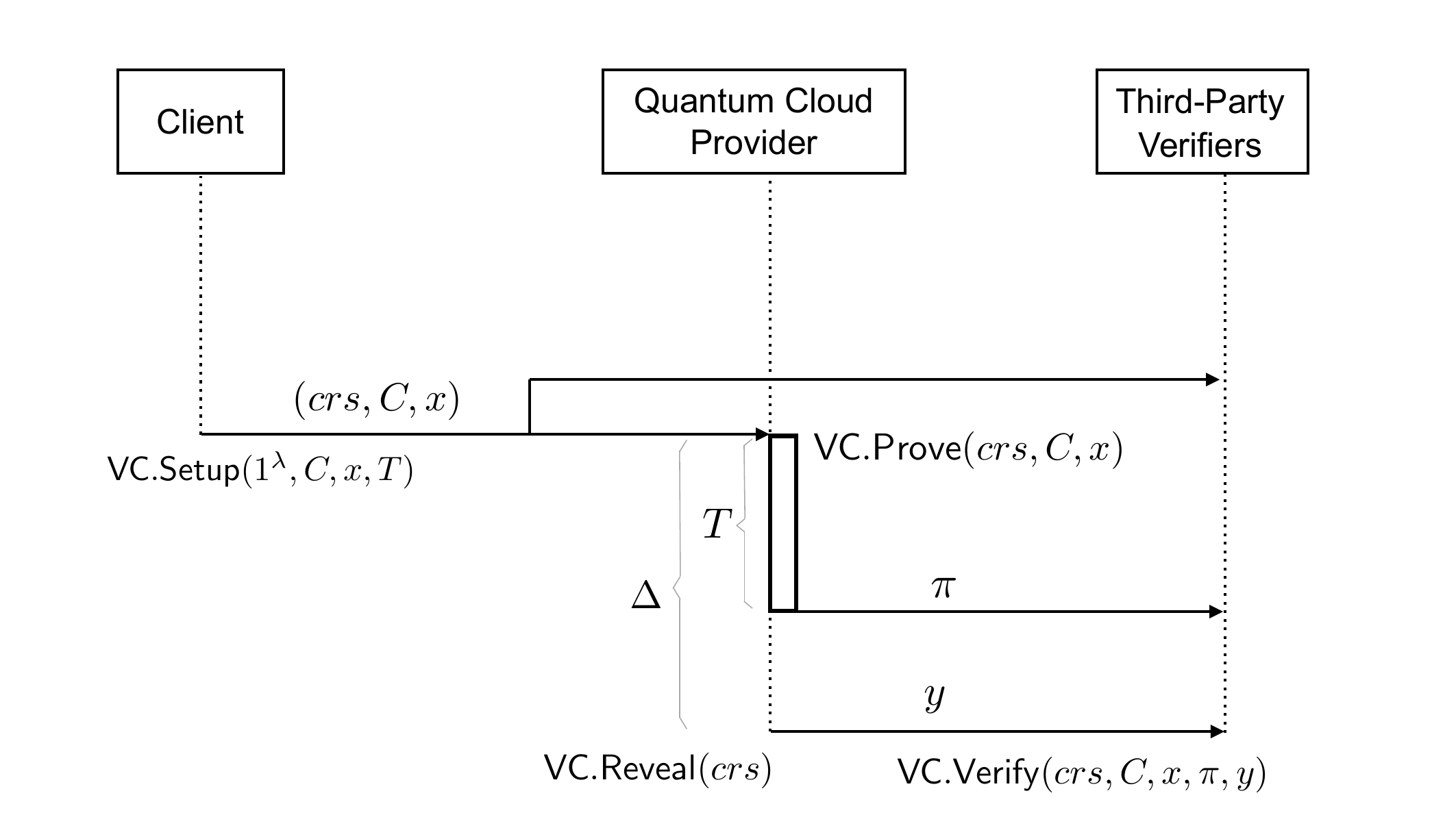}

    \caption{Outline of the post-quantum PV-CVQC protocol. The time to compute the delegated circuit is denoted by $T$. Thus the time to solve the puzzle should be set to some parameter $\Delta > T^{1+\epsilon}$ where $\epsilon > 0$.}
    \label{fig:overview-detailed}
\end{figure}

%% file: sections/table-results.tex
\begin{table*}[!ht]
\centering
\begin{tabular}{|c|c|c||c|c|c|c|c|}
\hline
\multicolumn{3}{|c||}{\textbf{Quantum Circuit}} & \multicolumn{5}{c|}{\textbf{Time-Lock Puzzle (TLP)}} \\
\hline
\textbf{\# Qubits} & \textbf{Depth} & \textbf{T (ms)} & \textbf{$\bm\TLPSolve$ (ms)} & \textbf{$\bm\TLPGen$ (ms)} & \textbf{$\bm\TLPSetup$ (ms)} & \textbf{Amortized $\bm\TLPSetup$ (ms)} & \textbf{$\mu$} \\
\hline
\hline
\multirow{6}{*}{5} & 10 & $1.84$ & $13.63$ & $0.05$ & $1445.57$ & $72.28$ & 1 \\
 & 20 & $2.31$ & $13.63$ & $0.05$ & $1445.57$ & $72.28$ & 1 \\
 & 50 & $4.56$ & $13.63$ & $0.05$ & $1445.57$ & $72.28$ & 1 \\
 & 100 & $8.13$ & $13.63$ & $0.05$ & $1445.57$ & $72.28$ & 1 \\
 & 200 & $16.84$ & $32.97$ & $0.05$ & $6577.50$ & $328.88$ & 3 \\
 & 300 & $22.84$ & $45.83$ & $0.06$ & $10513.51$ & $525.68$ & 4 \\
\hline
\multirow{6}{*}{10} & 10 & $3.14$ & $13.63$ & $0.05$ & $1445.57$ & $72.28$ & 1 \\
 & 20 & $5.34$ & $13.63$ & $0.05$ & $1445.57$ & $72.28$ & 1 \\
 & 50 & $11.78$ & $24.24$ & $0.06$ & $3672.94$ & $183.65$ & 2 \\
 & 100 & $20.47$ & $32.97$ & $0.05$ & $6577.50$ & $328.88$ & 3 \\
 & 200 & $41.06$ & $85.37$ & $0.05$ & $25915.06$ & $1295.75$ & 7 \\
 & 300 & $58.01$ & $117.25$ & $0.05$ & $48829.51$ & $2441.48$ & 10 \\
\hline
\multirow{6}{*}{15} & 10 & $10.52$ & $24.24$ & $0.06$ & $3672.94$ & $183.65$ & 2 \\
 & 20 & $19.39$ & $32.97$ & $0.05$ & $6577.50$ & $328.88$ & 3 \\
 & 50 & $49.51$ & $103.16$ & $0.05$ & $40500.14$ & $2025.01$ & 9 \\
 & 100 & $98.38$ & $196.48$ & $0.07$ & $120891.41$ & $6044.57$ & 17 \\
 & 200 & $168.56$ & $339.30$ & $0.05$ & $52020.00$ & $2601.00$ & 30 \\
 & 300 & $228.81$ & $462.57$ & $0.05$ & $97161.80$ & $4858.09$ & 41 \\
\hline
\end{tabular}%
\vspace{0.3cm}
\caption{The execution time (averaged over 20 trials) $T$ of running a random quantum circuit for different numbers of qubits and circuit depths and the corresponding time elapsed for the time-lock puzzle setup ($\bm\TLPSetup$, total and amortized over the 20 different puzzle generations), generation ($\bm\TLPGen$), and solving ($\bm\TLPSolve$ = $\Delta$). The parameter $\mu$ corresponds to the number of executions of the sequential function to achieve the corresponding $\Delta$ time for any given circuit execution time $T$.}
\label{tab:benchmark_results-}
\end{table*}

%% file: sections/05_experimental.tex
\section{Experimental Results}
\label{sec:experimental}

\input{sections/figure-graph-results}

In this section, we demonstrate the feasibility of our approach by simulating the protocol between a classical client who wishes to delegate quantum circuits of various sizes to a quantum server. Our implementation is a proof of concept with small security parameters, focusing primarily on measuring the time complexity of the parts of the protocol that dominate running time. In particular, we implemented the algorithms of the TLP $(\TLPSetup,\TLPGen,$ $ \TLPSolve)$, and the actual quantum circuit computation  executed by the prover, and which happens to be part of the $\VCProve$ subroutine. Other parts of the protocol such commitments and $\VCVerify$ were not simulated as they are relatively less time-consuming. The verifier's Feynman-Kitaev construction of the Hamiltonian~\cite{fitzsimons2018} is assumed to be part of the preprocessing phase before the protocol starts so we do not consider it in our simulation.

Our implementation includes a variant of the TLP puzzle proposed in \cite{agrawalr2024}, which is based on reusable garbling circuit techniques. While the original scheme uses lattice-based succinct randomized encodings for their reusable garbling mechanism, we instead build upon the implementation of \cite{harth-kitzerow2022} which exhibits weaker security guarantees but is more practical. Additionally, relying on the fact that random oracles are unconditionally sequential~\cite{mahmoody2013a}, we instantiate our sequential function using SHA-256. 

The simulations were performed on a t3.medium AWS instance with 2 vCPUs (3.1 GHz each), 4GB of memory, and access to an SV1 universal state vector simulator to execute quantum circuits (as the prover) on a cloud-based quantum simulation platform, specifically, AWS Braket. We use the Qiskit SDK~\cite{qiskit2024} to define, build, transpile, and run the quantum circuits on the SV1 simulator. Each instance of the simulation is averaged over at least 20 different executions of the quantum circuit. The source code is available at~\cite{QuantumVerificationTLP_Code}. 

\subsection{Random Quantum Circuits}

Table~\ref{tab:benchmark_results-} shows the results of our experiments when tested against randomly generated quantum circuits of 5, 10, and 15 qubits, with a depth ranging between 10 and 300 gates. Random circuits were sampled using Qiskit's \texttt{random\_circuit} function which generates circuits over a random distribution on the standard gates (e.g., Pauli, phase, swap, Hadamard, and controlled gates). For each qubit-depth setting, we simulated 20 different random quantum circuits and measured the time it takes for the prover to execute these circuits $T$. The client can use this to determine the required setting of $\Delta$ for any circuit of depth $D$ such that the time to solve the puzzle is much greater than $T$. In all cases, we observed that, for the verifier, generating a puzzle is almost constant, whereas the bulk of the time was spent in the setup phase and is attributed to the reusable garbling mechanism (which is dependent on the circuit depth).

Across all settings, the prover's circuit execution time 
$T$ increases with both depth and number of qubits (e.g. from 1.84\,ms for 5 qubits at depth 10, up to 228.81\,ms for 15 qubits at depth 300) which directly guides the client's choice of the time-delay parameter for releasing the verification  parameters. In particular, the puzzle-solving $\TLPSolve$ time is tuned to remain larger than the observed computation time. 
On the client side, puzzle generation $\TLPGen$ stays constant and negligible ($\approx$ 0.05--0.07\,ms). The dominant overhead appears in the puzzle setup cost $\TLPSetup$; however, this cost amortizes cleanly across multiple puzzle generations, as reflected by the amortized $\TLPSetup$ column.

\subsection{Enhanced Hybrid HHL}

We also evaluated the protocol over a variant of the Harrow–Hassidim –Lloyd (HHL) circuits \cite{harrow2009,yalovetzky2024}, specifically enhanced hybrid HHL, which are known to take advantage of quantum speedup when solving large (sparse) linear systems and have been found to be particularly effective in finance applications including mean-variance portfolio optimization and risk assessment~\cite{rebentrost2018}. They are thus appropriate candidates for testing as classical algorithms would want to delegate these circuits to quantum-capable, yet potentially untrustworthy servers. We used the enhanced hybrid HHL algorithm (which has both classical and quantum computations) as opposed to the vanilla one since the former is more efficient in practice given the current quantum hardware capabilities.

Figure~\ref{fig:hhl_results} shows the results for different sizes of circuits that correspond to different instances of the matrices which represent the system of linear equations. The depth can be used to estimate the time it takes to solve the circuit and thus determine the appropriate time $\Delta$ for which the puzzle should be solved. The parameter $\mu$ specifies the number of iterations of the sequential function to achieve this $\Delta$. We observe from Figure~\ref{fig:hhl_results}(a) that the circuit depth grows approximately linearly with the size of the problem instance $N$. Figure~\ref{fig:hhl_results}(b) indicates that the measured execution time increases substantially with matrix (instance) size $N$, reaching on the order of minutes for the largest tested instances. These results allow the client to scale the sequential work factor $\mu$ of the time-lock puzzle so that puzzle solving completes (in $\Delta$ time) only after the prover finishes the HHL computation and produces the timestamped proof. Lastly, we see from Figure~\ref{fig:hhl_results}(c) that the average fidelity (which roughly represents the correctness of the quantum circuit upon measurement) of HHL degrades almost linearly with $N$. This is expected, as deeper circuits are prone to decoherence due to noise. While quantum error correction can be used to mitigate this issue, the depth (and hence the running time $T$) would significantly increase and the client should take into account this increase by setting $\mu$ appropriately and get the desired $\Delta > T$.


%% file: sections/figure-graph-results.tex
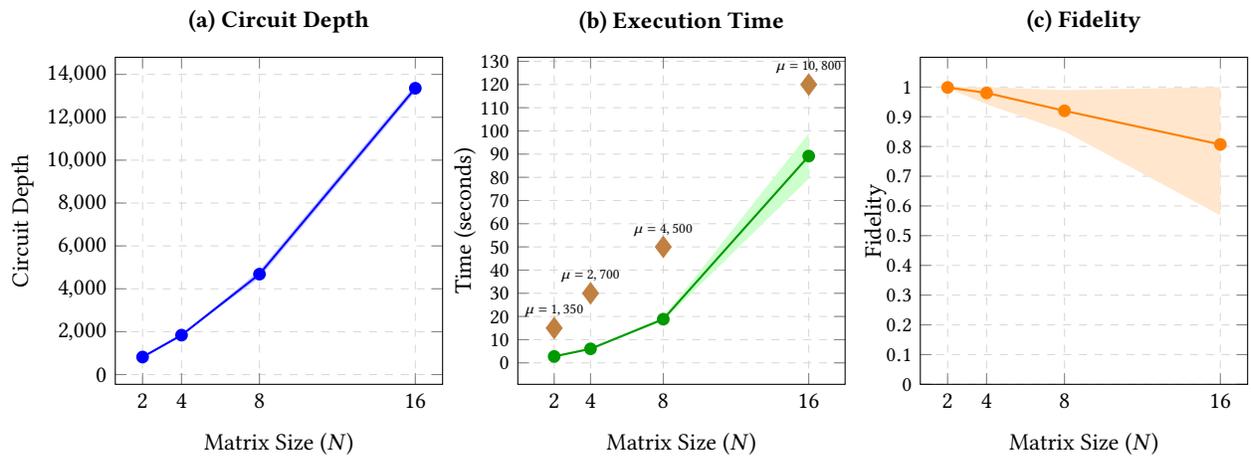
\begin{figure*}[h!]
\centering
\begin{tikzpicture}
\begin{groupplot}[
    group style={
        group size=3 by 1,
        horizontal sep=1cm,
    },
    width=0.7\columnwidth,
    height=0.7\columnwidth,
    grid=major,
    grid style={dashed, gray!30},
]

\nextgroupplot[
    xlabel={Matrix Size ($N$)},
    xtick = {2,4,8,16},
    ylabel={Circuit Depth},
    title={\textbf{(a) Circuit Depth}},
    legend pos=north west,
    ytick={0,2000,4000,6000,8000,10000,12000,14000},
    scaled y ticks=false,
]
\addplot[name path=upper,draw=none] coordinates { (2, 817.9) (4, 1839.4) (8, 4858.5) (16, 13527.0) };
\addplot[name path=lower,draw=none] coordinates { (2, 816.4) (4, 1838.8) (8, 4507.2) (16, 13180.3) };
\addplot[blue!20] fill between[of=upper and lower];
\addplot[blue,thick,mark=*] coordinates { (2, 817.1) (4, 1839.1) (8, 4682.9) (16, 13353.6) };

\nextgroupplot[
    xlabel={Matrix Size ($N$)},
    xtick = {2,4,8,16},
    ylabel={Time (seconds)},
    ylabel style ={yshift=-15pt},
    title={\textbf{(b) Execution Time}},
    legend pos=north west,
    ytick={0,10,20,30,40,50,60,70,80,90,100,110,120,130,140},
    yticklabel style={font=\small, /pgf/number format/.cd, fixed, 1000 sep={}},
    xmin=0, xmax=18,
]
\addplot[name path=upper,draw=none] coordinates { (2, 3.0375) (4, 6.3391) (8, 19.6287) (16, 98.9039) };
\addplot[name path=lower,draw=none] coordinates { (2, 2.4996) (4, 5.7629) (8, 18.0359) (16, 79.4126) };
\addplot[green!20,area legend] fill between[of=upper and lower];
\addplot[green!60!black,thick,mark=*] coordinates { (2, 2.7685) (4, 6.0510) (8, 18.8323) (16, 89.1582) };
\addplot[brown,only marks,mark=diamond*,mark size=4pt,forget plot] coordinates { (2, 15.00) (4, 30.00) (8, 50.00) (16, 120.00) };

\node[font=\tiny, yshift=1pt] at (axis cs:2,15.00) [above] {$\mu=1,350$};
\node[font=\tiny, yshift=1pt] at (axis cs:4,30.00) [above] {$\mu=2,700$};
\node[font=\tiny, yshift=1pt] at (axis cs:8,50.00) [above] {$\mu=4,500$};
\node[font=\tiny, yshift=1pt] at (axis cs:16,120.00) [above] {$\mu=10,800$};

\nextgroupplot[
    grid=major,
    grid style={dashed, gray!30},
    xlabel={Matrix Size ($N$)},
    xtick = {2,4,8,16},
    ylabel={Fidelity},
    ylabel style ={yshift=-18pt},
    title={\textbf{(c) Fidelity}},
    legend pos=north east,
    ymin=0,
    ytick={0,0.1,0.2,0.3,0.4,0.5,0.6,0.7,0.8,0.9,1},
    yticklabel style={font=\small},
]
\addplot[
    name path=upper,
    draw=none,
] coordinates {
    (2, 1.0000) (4, 1.0000) (8, 0.9895) (16, 1.0000)
};

\addplot[
    name path=lower,
    draw=none,
] coordinates {
    (2, 0.9962) (4, 0.9420) (8, 0.8506) (16, 0.5683)
};

\addplot[orange!20] fill between[of=upper and lower];

\addplot[
    orange,
    thick,
    mark=*,
] coordinates {
    (2, 0.9985) (4, 0.9801) (8, 0.9200) (16, 0.8068)
};

\end{groupplot}
\end{tikzpicture}

\caption{(a) Circuit depth scaling with matrix size, showing a linear relationship between the two. (b) The average running time of the HHL circuit for different matrix dimensions $N$ and the corresponding $\mu$ for the time-lock puzzle that ensures it is solved only after the HHL circuit is complete. (c) The average fidelity (how well it adheres to the true solution) for each matrix size ($N$)}
\label{fig:hhl_results}
\end{figure*}

%% file: sections/06_conclusion.tex
\section{Conclusion and Future Work}

This work proposes a verifiable computation scheme for quantum circuits which can be publicly verified by any third party after a set amount of time has elapsed. We demonstrate how a privately and classically verifiable quantum computation can be transformed into a publicly verifiable one by mainly relying on the recently introduced post-quantum time-lock puzzle \cite{agrawalr2024}.  Unlike previous work (that relied on iO), our scheme can be solely based on the LWE assumption, which is falsifiable and thus arguably more dependable than ones that are not falsifiable. 
%

It remains open to find an approach that achieves (non-time-based) public verifiability in the plain model and under standard assumptions. It would also be useful to have a publicly verifiable scheme that is applicable in the more common setting where the verification key is made available as soon as the circuit is delegated. Furthermore, the CRS in our scheme is not reusable (it can only be used to prove a single statement), regardless of whether the setup is instance-dependent or not. This is due to the fact that security (verifiability) is only guaranteed as long as the time-lock puzzle, which hides the secret verification key, is not solved. Therefore, it is an interesting open question to upgrade this scheme into one with a universal CRS that can be used to prove the correctness of multiple computations at once without increasing its size. Lastly, while verifiable computation schemes are generally concerned with malicious provers, extending security against malicious verifiers, who may attempt to falsely accuse the prover of providing an incorrect proof, is a promising future direction.



%% file: sections/08_acknowledgments.tex
\section*{Acknowledgments}
This work was funded by Kuwait Foundation for the Advancement of Sciences (KFAS) under project code: PA24-6TE-2493.

%% file: sections/Appndx_commitment_def.tex
\section{Formal Definition of Commitment Scheme}\label{sec::Formal-Definition-of-Commitment Scheme}

In this section, we present formal definition of a commitment scheme. 

\begin{definition} 
\label{def:commitment}
A commitment scheme consists of two PPT algorithms $(\comcom,\comver)$ with the following properties:
\begin{itemize}
    \item Perfect completeness: For any $m$ and $r$:
    \begin{align*}
        \comver(\comcom(m,r),(m,r)) = 1
    \end{align*}    
    \item Computationally binding: For any QPT adversary $A$ that outputs $(m_0,r_0,m_1,$ $r_1,c)$ such that $m_0 \neq m_1$ and $|m_0| = |m_1|$ the following holds for sufficiently large $\secp$:
    \begin{equation*}
        \Pr\left[
            \begin{array}{c}
            \forall b \in \{0,1\}, \comver(c, (m_b,r_b)) = 1  
            \end{array}
            \right]
            \leq \negl(\secp)
    \end{equation*}
    \item Statistically hiding: For any (potentially unbounded) adversary $A$ the following holds for sufficiently large $\secp$:
        \begin{align*}
            \Pr\left[
            \begin{array}{c}
            m_0 \neq m_1 \\
            |m_0| = |m_1| \\
            b' = b
            \end{array}
            \middle| 
            \begin{array}{c}
            (m_0,m_1,st) \gets A(1^\secp) \\
            b \xleftarrow{\$} \{0,1\} \\
            c \gets \comcom(m_b,r) \\
            b' \gets A(st, c)
            \end{array}
            \right]
            \leq \dfrac{1}{2} + \negl(\secp)
        \end{align*}
\end{itemize}
\end{definition}